\documentclass[10pt, a4paper]{article}
\usepackage{xcolor}
\usepackage{graphicx} 
\usepackage{caption}
\usepackage{multirow}
\usepackage{amsmath} 
\usepackage{amssymb} 
\usepackage[figuresright]{rotating}
\usepackage{authblk}
\usepackage{booktabs}
\usepackage{array}
\usepackage{multicol}
\usepackage{todonotes}
\usepackage{amsthm}
\usepackage{textcomp}
\usepackage{subcaption,xcolor,lipsum}

\usepackage{algorithm}
\usepackage{algpseudocode}
\algdef{SE}[SUBALG]{Indent}{EndIndent}{}{\algorithmicend\ }%
\algtext*{Indent}
\algtext*{EndIndent}

\usepackage[symbol]{footmisc}

\newtheorem{defi}{Definition}
\newtheorem{theo}{Theorem}
\newtheorem{prop}{Proposition}
\newtheorem{cor}{Corollary}
\newtheorem{lem}{Lemma}
\newtheorem{example}{Example}

\newcommand\Tstrut{\rule{0pt}{2.6ex}}

\DeclareMathOperator{\Haf}{Haf}
\DeclareMathOperator{\PMP}{PMP}
\DeclareMathOperator{\diag}{diag}
\DeclareMathOperator{\arctanh}{arctanh}
\DeclareMathOperator{\Perf}{Perf}

\title{A Quantum Photonic Approach to \\ Graph Coloring}
\author[1]{Jesua Epequin\footnote{Corresponding author. Email: jesua.epequin@edf.fr}}
\author[2]{Pascale Bendotti}
\author[2]{Joseph Mikael\footnote{On leave from EDF and now with QUOBLY}}
\affil[1]{\small{\textit{EDF (China) Holding Ltd., 12 Ding Jianguomenwai Ave., 100022 Beijing, China}}}
\affil[2]{\textit{EDF Lab Paris-Saclay, 7 Bd. Gaspard Monge, 91120 Palaiseau, France}}
\date{\today}

\begin{document}
\maketitle
\begin{abstract}
Graph coloring is a central problem in combinatorial optimization. In this work, we investigate a hybrid classical-quantum heuristic that uses the ability of Gaussian Boson Sampling (GBS) to sample dense subgraphs and integrate this mechanism into an independent-set formulation of graph coloring. Through a polynomial graph transformation, color classes are represented as independent sets, which correspond to cliques in the complement of an augmented graph and can therefore be targeted using GBS-based sampling. The resulting Gaussian Boson Sampling Coloring (GBSC) algorithm is benchmarked against classical heuristics and an exact algorithm on two sets of instances: Erd\"os-R\'enyi random graphs and graphs derived from a smart-charging application. The experiments, performed by classically simulating the GBS sampling process, show that GBSC produces competitive colorings and performs particularly well on dense graph instances. These results illustrate a specific use case in which GBS-based sampling can be integrated into a classical combinatorial optimization heuristic.


\end{abstract}

\section{Introduction}
Graph coloring is used as a paradigm for different applications in operations research such as the frequency assignment problems~\cite{HaleFreqAssign}, the register allocation problem~\cite{ChowHenRegisterAlloc} and the timetabling problem~\cite{DewerraTimeTab}.
Due to the computational complexity of the graph coloring problem, heuristic methods are often required to obtain solutions, especially for large-scale graphs. 
Although heuristics do not guarantee optimal solutions, it is of significant interest to design heuristics providing high-quality approximations for specific classes of instances.


The seminal work of \cite{aaronson2011computational} provided a theoretical foundation showing that a model of quantum computation-based solely on linear optics-can solve certain sampling problems, classically intractable under plausible complexity assumptions. Among these problems is Gaussian Boson Sampling (GBS), which lies within the $\#P$ complexity class. Experimental realizations of GBS have demonstrated a quantum advantage over classical methods, initially on static Gaussian boson distributions \cite{zhong2020quantum}, and later on parametrized distributions \cite{zhong2021phase, madsen2022quantum}. A key feature of GBS is its ability to encode the parameters of a Gaussian boson distribution into the adjacency matrix of a graph. Sampling from such distributions enables the identification of dense subgraphs with high probability \cite{PhysRevLett.121.030503}, utilizing the capacity of the quantum device to perform computations that are computationally intensive for classical machines.


In this paper, we employ GBS to tackle the graph coloring problem using an independent set formulation. By exploiting the equivalence between independent sets and cliques in the complement graph, we enable the use of GBS to identify promising colorings.  
We evaluate the proposed approach on two classes of instances. First, we consider Erdős--Rényi random graphs covering a range of graph densities and chromatic numbers. Second, we consider graphs derived from a smart-charging application, in which the number of charging terminals can be formulated as a graph-coloring objective. We compare GBSC with three established coloring heuristics, namely Smallest Last with Interchange (SLI), Recursive Largest First (RLF), and Degree of Saturation (Dsatur), and use exact chromatic numbers obtained from a branch-and-price algorithm as a reference.

The main contribution is to devise the GBS Coloring algorithm, a hybrid quantum-classical heuristic, to solve the graph coloring problem and evaluate its comparative performance with standard heuristics.

The remainder of this paper is organized as follows. Section \ref{sec:GBS} introduces Gaussian Boson Sampling and its relevance to graph-based problems. Section \ref{sec:graph_coloring} discusses the graph coloring problem, and reviews several solution strategies, including classical exact and heuristic methods, together with the novel hybrid classical-quantum approach integrating GBS. Section \ref{sec:results} presents numerical experiments that compare our method with established heuristics. Finally, Section \ref{sec:conclusion} summarizes our findings and outlines future research directions.

\section{Gaussian Boson Sampling}\label{sec:GBS}
Quantum algorithms are often studied in contexts where their intrinsic randomness is seen as a challenge rather than an advantage, as demonstrated in Shor's algorithm. However, when applied to the appropriate problems, this randomness can be leveraged as a powerful tool. One such example is Gaussian Boson Sampling (GBS), a specialized photonic protocol that takes advantage of quantum randomness for specific tasks. In this section, we introduce GBS and highlight a potential application.

The detailed mathematical formulation of the GBS sampling distribution, including the role of Hafnians, and a discussion on complexity, is provided in Appendix~\ref{app:GBS}.

\subsection{Graph Encoding}\label{GBSGraph}


Gaussian Boson Sampling (GBS) is a photonic quantum computing model in which squeezed states of light are injected into a linear-optical interferometer and subsequently measured. A \textit{squeezed state} is a non-classical quantum state of light in which the quantum uncertainty of one field quadrature is reduced below the vacuum level at the expense of increased uncertainty in the conjugate quadrature, while preserving the Heisenberg uncertainty relation. The interference of these squeezed states within the optical network generates highly correlated photon distributions whose sampling is computationally difficult for classical computers. As a result, GBS has emerged as a promising framework for addressing combinatorial optimization problems and for demonstrating quantum computational advantage.

The optical system is described in terms of quantum optical modes, or \emph{qumodes}, each of which can contain an arbitrary number of photons. We denote a state of $k$ single-photons arranged in $n$ qumodes as $\mathbf{s}=\lvert s_1,\ldots,s_n\rangle$, where $s_i$ is the number of photons in the $i$-th qumode and $\sum_i s_i=k$ is the total number of photons.

Let $G=(V,E)$ be an undirected graph, with vertex set $V=\{v_1,\ldots,v_n\}$ and edge set $E$. Its \emph{adjacency matrix} $A$ is defined by $A_{ij}=1$ if $(v_i,v_j)\in E$ and $A_{ij}=0$ otherwise. Such a matrix completely encodes the graph $G$ and is symmetric by definition. A subgraph of $G$ induced by a subset of its vertices can be represented by a binary sequence $\{x_1,\ldots,x_n\}$, where $x_i\in\{0,1\}$ indicates whether vertex $v_i$ is included in the subgraph. The adjacency matrix $A_{\mathbf{x}}$ of this subgraph, denoted $G[\mathbf{x}]$, is obtained from $A$ by deleting the rows and columns corresponding to vertices for which $x_i=0$. In this way, a binary photon-number qumode $\mathbf{s}=\lvert s_1,\ldots,s_n\rangle$ can be identified with an induced subgraph $G$ and its adjacency matrix. 

A \emph{perfect matching} of a graph is a set of edges such that every vertex belongs to exactly one edge of the set. Let $\Perf(G)$ denote the number of perfect matchings of $G$. As shown in Appendix~\ref{app:GBS}, the probability of sampling the subgraph $G[\mathbf{s}]$ (or equivalently, the qumode $\mathbf{s}=\lvert s_1,\ldots,s_n\rangle$) is given by
\begin{align}\label{GBSSubGraphProb}
P(G[\mathbf{s}]) \propto \frac{\Perf(G[\mathbf{s}])^2}{s_1!\ldots s_n!},
\end{align}

The correspondence between the number of perfect matchings in a graph and its \emph{density}, defined as $d(G)=\frac{2\lvert E\rvert}{\lvert V\rvert(\lvert V\rvert-1)}$, is highlighted in \cite{Aaghabali}. Intuitively, graphs with many edges tend to admit more perfect matchings. More precisely, it was established that the number of perfect matchings in a graph $G$ with an even number of vertices is upper-bounded by an increasing function of the number of edges. Therefore, the number of perfect matchings is positively correlated with the density of a graph. This implies that GBS devices sample high density subgraphs with high probability.

\subsection{Application to Maximum Clique}\label{GBS_clique}
A clique in a graph is a subgraph in which every pair of distinct vertices is adjacent. It can be described as a complete graph with density equal to one. 
The maximum clique problem involves finding the largest subgraph, $\mathcal{S}$ that satisfies the following condition of unit density:
$$
argmax_\mathcal{S}\{\lvert S\rvert:d(\mathcal{S})=1\}
$$
The maximum clique problem is known to be NP-hard, meaning that there is no known polynomial-time algorithm to solve it in the general case. Its applications cover a large range of fields, including computational biology \cite{Malod-Dognin}, social network analysis \cite{Balasundaram}, and scheduling problems (as presented below).
 
In classical computing, a common heuristic approach to finding maximum cliques in large graphs involves local search techniques \cite{Pullman}. These methods are favored for their ability to efficiently explore the solution space.  While they do not guarantee identification of the global maximum clique, they often succeed in finding large cliques in polynomial time. This makes them practical for use in real-world applications.

The corresponding algorithms start with a clique $\mathcal{C}$ and rely on two key operations: growing and swapping.
The growing operation is to expand clique $\mathcal{C}$ by adding one vertex at a time. The algorithm selects in the remainder of the graph a vertex that is connected to every vertex in $\mathcal{C}$, and appends it, thus resulting in a larger clique. The process continues until no additional vertices can be added without violating the clique condition. In the swapping operation, the algorithm finds the set of vertices in the remainder of the graph that are connected to all but one vertex from $\mathcal{C}$, and swaps one of these vertices with the corresponding vertex in the clique. Both operations are repeated until no further growth can occur.

As mentioned above, GBS devices can be programmed to sample dense subgraphs with high probability. Since cliques are graphs of unit density, they are the most likely subgraphs to be sampled from GBS. A way to improve local search algorithms is to sample dense subgraphs from GBS and use them as starting point for grow and swap algorithms. Not all sampled subgraphs are guaranteed to be cliques, but they can be obtained from dense subgraphs by removing vertices with low degree until a clique is found.  

\section{Graph coloring}\label{sec:graph_coloring}
Given a graph $G=(V,E)$ with $n$ vertices and $m$ edges, a \emph{graph coloring} is an assignment of labels called colors to each vertex of the graph so that no two adjacent vertices are of the same color. A coloring using at most $k$ colors is called a $k$-coloring, and can be defined as a function $c:V\rightarrow\{1,2,\ldots,k\}$ such that $c(u)\neq c(v)$, whenever $(u,v)\in E$. The smallest number of colors needed to color a graph $G$ is called its \emph{chromatic number}, and is often denoted by $\chi(G)$. A minimum coloring is a coloring with the fewest different colors. The graph coloring problem is to find a minimum coloring. For $k\ge 3$ the graph coloring problem is NP-hard in the general case.

An \emph{independent set} of $G$ is a subset of $V$ such that no two vertices are adjacent. Note that in any coloring of $G$ all vertices with the same color constitute an independent set, then a $k$-coloring is a partition of $V$ into $k$ independent sets, namely $\{V_1, \ldots, V_k\}$. A \emph{maximal independent set (MIS)} is an independent set not strictly included in any other. 

Cliques and independent sets are complementary concepts. This relationship allows for the translation of problems and solutions between the two concepts. The translation involves the complement of a graph $G$.
\begin{defi}
The complement of $G$ is the graph $\bar{G}=(\bar{V},\bar{E})$, with same vertex set as $G$, namely $\bar{V}=V$, and whose edge set is defined as
$$
\bar{E}=\{(u,v)\in V\times V\hspace{1pt}\lvert\hspace{1pt}(u,v)\notin E\}.
$$
\end{defi}
From this definition it follows that the MIS problem in $G$ is equivalent to the maximum clique problem in $\bar{G}$. 
\begin{prop}\label{prop:mis_clique}
The set of vertices that form a maximum clique in $\bar{G}$ corresponds to vertices in $G$ that constitute a maximum independent set.
\end{prop}

\subsection{Reduction to maximum clique}\label{sec:reduction_maxclique}
The $k$-coloring problem involves assigning one of $k$ possible colors to each vertex of a graph so that no two adjacent vertices share the same color. This classical problem can be related to  finding a maximum clique in a suitably constructed graph. 
\begin{defi}
The augmented $k$-graph $G_k=(V_k,E_k)$, has a vertex set defined as
$$
V_k = \{v^i\hspace{1pt}\lvert\hspace{1pt}v\in V,\hspace{1pt} i\in\{1,\ldots,k\}\}.
$$
The edge set is the union $E_k = E^{\text{\textdblhyphen}} \cup E_{\text{\textdblhyphen}}$ where
\begin{align*}
E^{\text{\textdblhyphen}} & = \{(u^i,v^i)\hspace{1pt}\lvert\hspace{1pt}(u,v)\in E,\hspace{1pt} i\in\{1,\ldots,k\}\} \hspace{1pt}  \\
E_{\text{\textdblhyphen}} & =\{(v^i,v^j)\hspace{1pt}\lvert\hspace{1pt}v\in V,\hspace{1pt}i \neq j\in\{1,\ldots,k\}\}
\end{align*}
\end{defi}

The augmented $k$-graph $G_k$ is built by creating $k$ copies of each vertex in the original graph $G$, one for each possible color. Edge set $E_k$ is constructed to reflect two kinds of constraints. Set $E^{\text{\textdblhyphen}}$ includes conflict edges: they connect vertices $u^i$ and $v^i$ if $u$ and $v$ are adjacent in the original graph and are both assigned the same color $i$, a direct violation of the coloring constraint. Meanwhile, set $E_{\text{\textdblhyphen}}$ includes clonal edges, which connect all copies $v^i$ and $v^j$ of the same vertex $v$ across different colors. This ensures that only one copy of each vertex can appear in any independent set. It is worth noting that this transformation is polynomial. 

\begin{prop}
The transformation from $G$ to $\bar{G}_k$ is polynomial in the size of the instance $G$ and the number of colors $k$. 
\end{prop}
\begin{proof}
The statement is derived from the fact that $G_k$ (resp. $\bar{G_k}$) contains $kn$ vertices and $km+nk(nk-1)/2$ edges (resp. $nk^2(n-1)/2-km$ edges).
\end{proof}

\begin{defi}\label{def:projection}
Let $H_k$ be a subgraph of $G_k$, the \emph{projection} $p_G:V(H_k)\rightarrow V$ of $H_k$ on $G$ is the function defined by $p_G(v^i)=v$ for $v^i\in H_k$.
\end{defi}

When there is no risk of confusion, we will denote the image of the projection $p_G:V(H_k)\rightarrow V$ by $p_G(H_k)$. 

\begin{lem}\label{lem:injective_projection}
If $H_k$ is an independent set of $G_k$, then $p_G:V(H_k)\rightarrow V$ becomes an injective function.
\end{lem}
\begin{proof}
Let $u^i$ and $v^j$ be two vertices in $H_k$ with the same projection to $V$. By Definition \ref{def:projection} this means that $u=v$. If $i\neq j$ then $(v^i,v^j) \in E_{\text{\textdblhyphen}}$, contradicting the independence of $H_k$. The lemma follows.   
\end{proof}

\begin{theo}\label{theo:kcoloring_kclique}
The graph $G$ admits a $k$-coloring if and only if the complement graph $\bar{G}_k$ contains a clique of size $\lvert V\rvert$.
\end{theo}
\begin{proof}
$(\Rightarrow)$ Let $c:V\rightarrow\{1,2,\ldots,k\}$ be a $k$-coloring of $G$. From Proposition \ref{prop:mis_clique}, a clique in the complement graph $\bar{G}_k$ corresponds to an independent set in $G_k$. Consider the subgraph $I_k$ of $G_k$ induced by the vertex set $\{v^{c(v)}\hspace{1pt}\lvert\hspace{1pt} v \in V\}$, and let $u^{c(u)}$ and $v^{c(v)}$ be two different vertices in $I_k$. By construction $u\neq v$, which implies that $(u^{c(u)},v^{c(v)})\notin E_{\text{\textdblhyphen}}$. Furthermore, if $(u,v)\in E$, then $c(u)\neq c(v)$, because $c$ is a coloring, whence $(u^{c(u)},v^{c(v)})\notin E^{\text{\textdblhyphen}}$. Therefore, there is no edge between any pair of vertices in $I_k$, implying that it is an independent set in $G_k$. Since, $I_k$ contains exactly one vertex for each $v\in V$, it has size $\lvert V\rvert$, completing the proof.

$(\Leftarrow)$ Let $\mathcal{C}_k$ be a clique in $\bar{G}_k$, and $I_k$ the corresponding independent set in $G_k$. Let $c:p_G(I_k)\rightarrow\{1,2,\ldots,k\}$ be a relation taking values on the projection $p_G(I_k)\subset V$ of $I_k$ in $G$, and defined by $c(v)=i$, when $v^i\in I_k$. If $c$ assigns two different values $i\neq j$ to a given $v\in p_G(I_k)$, by definition $v^i$ and $v^j$ belong to $I_k$, which contradicts the independence of $I_k$, since in this case $(v^i,v^j)\in E_{\text{\textdblhyphen}}$. This implies that $c$ is indeed a function. If $u$ and $v$ are adjacent in $G$, since $(u^{c(u)},v^{c(v)})\notin E^{\text{\textdblhyphen}}$, then $c(u)\neq c(v)$, which ensures that $c$ defines a coloring of $p_G(I_k)$. Finally, from Lemma \ref{lem:injective_projection}, the function $p_G:V(I_k)\rightarrow V$ is injective. If $\mathcal{C}_k$ has size $|V|$, then so does $I_k$ and this function becomes bijective, which implies $p_G(I_k)=V$. The result follows. 
\end{proof}

\begin{figure}
    \centering
    \includegraphics[width=\linewidth]{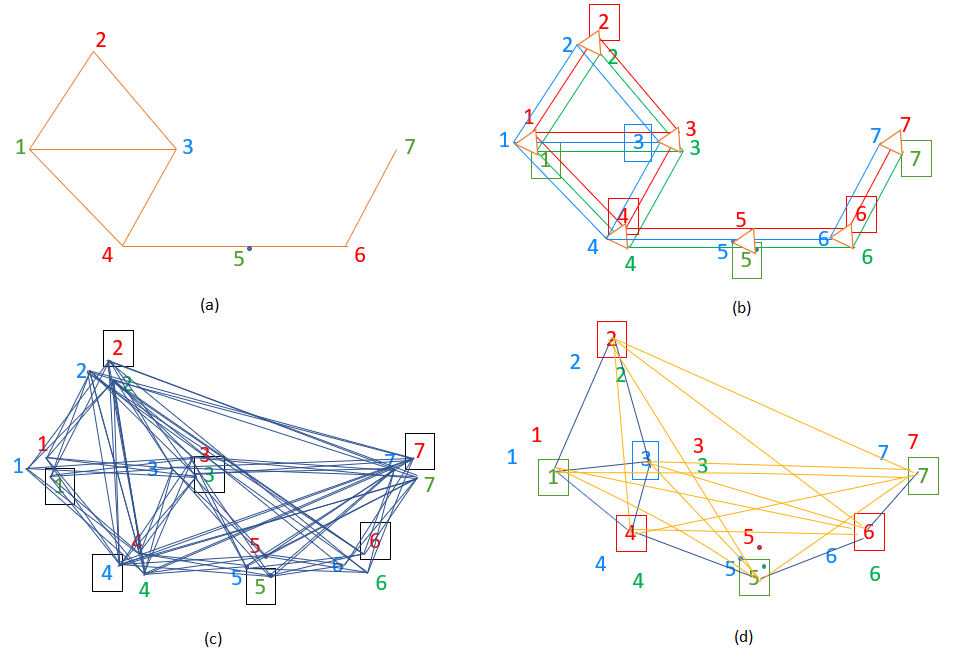}
    \caption{
    An illustration of Theorem \ref{theo:kcoloring_kclique} with a 7-vertex graph example}
    \label{fig:MIS}
\end{figure}

\begin{example}
Figure \ref{fig:MIS} illustrates Theorem \ref{theo:kcoloring_kclique} through a simple example. The original graph $G$ is presented in (a) together with a valid 3-coloring. The augmented graph $G_3$ is shown in (b), and its complement $\bar{G}_3$ in (c). Finally, in (d) a clique of size $|V|$ in $\bar{G}_3$ is displayed, which directly corresponds to the given $3$-coloring of $G$. The blue edges emphasize how the structure of $G$ is preserved in the clique representation.  
\end{example}

The proof of Theorem \ref{theo:kcoloring_kclique} has the following consequence.

\begin{cor}
Let $\mathcal{C}_k$ be a maximum clique in $\bar{G}_k$, and let $H$ be the subgraph of $G$ induced by the vertex set $p_G(\mathcal{C}_k)$ (Definition \ref{def:projection}). Then $\mathcal{C}_k$ defines a coloring  $c:H\to\{1,\ldots,k\}$  given by $c(v)=i$ whenever $v^i\in\mathcal{C}_k$
\end{cor}

Let again $\mathcal{C}_k$ be a clique in $\bar{G_k}$.
If $\lvert\mathcal{C}_k\rvert=\lvert V\rvert$ then $\mathcal{C}_k$ provides a coloring of $G$. Otherwise, if $\lvert\mathcal{C}_k\rvert<\lvert V\rvert$, at least one additional color is required to color the vertices not in $\mathcal{C}_k$. If the vertices with no color can be colored with one additional color, a $k+1$-coloring has been found. Otherwise $k$ is increased to solve another iteration of the $k$-coloring problem.

\subsection{Heuristic algorithms} \label{sec:heur}
This section describes three classical heuristic graph coloring methods.

\emph{Degree of saturation (Dsatur) \cite{Brelaz}.} This algorithm generates a vertex ordering and then colors the graph vertices one after another, adding a previously unused color when needed. Its goal is to prioritize the vertices that are seen to be the most constrained. These are dealt with first, allowing the less constrained vertices to be colored later. It starts by coloring the vertex with the maximal degree. The algorithm then determines which of the remaining uncolored vertices has the highest \emph{degree of saturation}, i.e., the highest number of colors in its neighborhood; and colors this vertex next. In cases of ties, the algorithm chooses the vertex with the largest degree in the subgraph induced by the uncolored vertices. 

\emph{Recursive Largest First (RLF) \cite{Leighton}.} This algorithm assigns colors to vertices in a graph by building each color class -- a set of vertices which can be colored with a single color -- one at a time. It iteratively selects a vertex for coloring, which will, in some sense, leave the resulting uncolored vertices colorable in as few colors as possible. More precisely, it starts by assigning a color to the vertex of maximal degree. Subsequent vertices assigned to the same color are chosen as those that (a) are not currently adjacent to any colored vertex and (b) have a maximal number of neighbors that are adjacent to colored vertices. Ties can be broken by selecting the vertex with the minimum degree in the subgraph of uncolored vertices. Vertices are selected in this way until it is no longer possible to add further vertices. Then the entire process is repeated recursively on the subgraph of uncolored vertices using the next available color.

\emph{Smallest Last with Interchange (SLI)} is a sequential coloring method introduced in \cite{Matula}, based on the idea that vertices with fewer neighbors should be colored later. The algorithm creates an ordering ${v_1,\ldots,v_n}$ of the vertices in which $v_n$ has the smallest degree, and for each $i = 1, \dots, n-1$, $v_i$ has the smallest degree in the subgraph induced by $\{v_{i+1},\ldots,v_n\}$. When coloring a vertex, SLI may need to introduce a new color. However, it tries to avoid this by performing an \emph{interchange}. Specifically, if $v_m$ needs to be colored and colors $1$ through $k$ have already been used, we examine subgraphs $G_{ij}$ formed by the vertices colored with color $i$ or $j$. If we can find colors $i$ and $j$ such that no connected component of $G_{ij}$ has two differently colored vertices, both adjacent to $v_m$, we can perform an interchange on each connected component of $G_{ij}$ which contains an $i$-colored vertex adjacent to $v_m$. This allows $v_m$ to be assigned color $i$ without introducing a new color. If no such choice exists, color $k+1$ must be used for $v_m$.

\subsection{Exact algorithm}\label{sec:exact}
This section gives a brief description of the exact algorithm used for the numerical experiments in Section~\ref{sec:results}. 

One of the best methods known for determining lower bounds on the vertex coloring number of a graph is the linear-programming column-generation
technique, where variables correspond to independent sets~\cite{Mehrotra}. 

For all maximal independent sets $S\in\mathcal{S}$, the optimal chromatic number $\chi(G)$ is given by the following integer-programming problem
\begin{eqnarray}
\chi(G) =  &\min  \quad 
\sum_{S\in\mathcal{S}} x_S & \label{obj}\\
&\sum_{S\in\mathcal{S}:v\in S} x_S\geq 1 &\forall v\in V \label{ineq:covering}\\
&x_S\in \{0,1\} & \forall S\in \mathcal{S} \label{varInt}
\end{eqnarray}

The objective function~\eqref{obj} corresponds to the number of colors used in the solution. The covering constraint~\eqref{ineq:covering} ensures that each vertex of the graph belongs to at least one independent set. This formulation is not compact as the number of constraints is in $O(n)$ and the number of variables in $O(2^n)$.

When handling such an exponential formulation in the number of variables, a branch-and-price method is required. It amounts to a branch-and-bound method in which at each node of the search tree, columns may be added to the linear programming relaxation.  This approach was proposed in its main components for  formulation~\eqref{obj}--\eqref{varInt} in~\cite{Mehrotra}. In~\cite{Held} a complementary work dealt with numerical difficulties
in the context of column generation. The authors devised a complete and efficient branch-and-price algorithm using their
own branch-and-bound based algorithm for solving the pricing subproblem \footnote{Code available at github.com/heldstephan/exactcolors.}.    

\begin{algorithm}
\caption{\texttt{GBSC}}
\begin{algorithmic}[1]
\Require Graph $G=(V,E)$
\Ensure Color assignment \texttt{color} for all $v \in V$
\State Initialize uncolored graph $H \gets G$
\State Initialize $\texttt{color}(v) \gets |V|$ for all $v \in V$
\While{there is $(u,v)$ in $E$ with $\texttt{color}(u) = \texttt{color}(v)$}
    \State $C_{\text{max}} \gets \texttt{FindCliques}(H)$ (See Alg. 2)
    \State $C \gets \texttt{BestClique}(H,C_{\text{max}})$ (See Alg. 3)
    \ForAll{vertices $v^i$ in clique $C$}
        \State $\texttt{color}(v) \gets i$
        \State $H\gets H\setminus\{v\}$
    \EndFor
\EndWhile
\State \Return $\texttt{color}$
\end{algorithmic}
\end{algorithm}

\subsection{GBS coloring algorithm}
In this section, we give a detailed explanation of our coloring method, which is divided into three algorithms. Algorithm 1 introduces the coloring procedure, Algorithms 2 and 3 are accessory and are presented separately for better readability.
The latter algorithms make also use of subroutines \texttt{sample}, \texttt{shrink}, and \texttt{search} provided by the Strawberryfields library, a short description of which is given in Appendix~\ref{app:strawberryfields}.
An illustration of the proposed method is shown in the following example.
\begin{example} \label{ex2}
Figure~\ref{fig:flowchart} illustrates  the iterative principle of the proposed algorithm \texttt{GBSC} through an example. 
The first iteration starts from the input graph. This graph is transformed into its $k$-augmented version (here $k=3$ is the ceiling of the Hoffman bound), and then into the complemented augmented graph. The central step is the identification of cliques via \texttt{FindCliques}, with the most significant one (\texttt{BestClique}) highlighted in red. The resulting $k$-coloring colors all vertices except $\{1,2,8\}$. Since uncolored vertices remain, a second iteration is performed on their induced subgraph, following the same steps (here with $k=2$). The final 5-coloring of the original graph, obtained by combining both iterations, is shown at the bottom right.
\end{example}

\emph{Algorithm 1}. To compute a valid coloring for a graph $G=(V,E)$, we introduce algorithm \texttt{GBSC} (Gaussian Boson Sampling Coloring). It starts by initializing the \emph{residual uncolored graph} $H$ (the subgraph of $G$ consisting of uncolored vertices) as $G$ (line~1). Each vertex $v \in V$ is initially assigned the placeholder color $|V|$ (line~2), ensuring that all vertices start in the same color class. The algorithm iteratively resolves conflicts, while there exists an edge in $E$ whose ends are assigned to the same color (line~3), it proceeds as follows: 

- Call Algorithm 2 (\texttt{FindCliques}) on the current residual graph $H$ (line~4) to generate a list of candidate cliques using Gaussian Boson Sampling (see Section \ref{GBS_clique} for details).  

- Call Algorithm 3 (\texttt{BestClique}) from the generated list to select the best clique  according to a hierarchical set of criteria (line~5). 

Once the best clique $C$ is selected, for each vertex $v^i$ in $C$ (line~6), 
color $i$ is assigned to $v$ (line~7) and 
vertex $v$ removed from $H$ (line~8), thus reducing the residual graph size. The \emph{while} loop continues until there are no conflicting edges, at which point all vertices are properly colored. The algorithm then returns the final color assignment (line~11). 

\begin{algorithm}
\caption{\texttt{FindCliques}}
\begin{algorithmic}[1]
\Require Graph $H$
\State $h \gets$ Hoffman bound of $H$
\State $k \gets \lceil h \rceil$
\State $H_k\gets$ augmented graph of $H$ 
\State $\bar{H_k} \gets$ complement graph of $H_k$
\State $S_{\text{list}} \gets$ \Call{\texttt{sample}}{$\bar{H_k}, |H|$} 
\State $C_{\text{start}} \gets$ [\Call{\texttt{shrink}}{$S$} : $S \in S_{\text{list}}$] 
\State $C{_\text{max}} \gets$ [\Call{\texttt{search}}{$C,\bar{H_k}$} : $C \in C_{\text{start}}$] 
\State \Return $C_{\text{max}}$
\end{algorithmic}
\end{algorithm}

\emph{Algorithm 2}. \texttt{FindCliques} is introduced to find a list of candidate cliques in the complement of an augmented graph of an input graph $H$. It starts by calculating the Hoffman lower bound $h$ for the chromatic number of $H$ \cite{hoffman} (line~1), and constructs the augmented graph $H_k$, where $k = \lceil h \rceil$ (lines~2, 3). It then proceeds to calculate the complement graph $\bar{H_k}$ (line~4). As described in Section \ref{sec:reduction_maxclique}, the $k$-coloring problem can be reduced to finding a maximum independent set in $H_k$, which is equivalent to identifying a maximum clique in $\bar{H_k}$. In line~5 we use Gaussian Boson Sampling (subroutine \texttt{sample}) to obtain dense subgraphs in $\bar{H_k}$ with mean photon number (vertex count) equal to $|H|$. In line~6 we call subroutine \texttt{shrink} to remove vertices with low degree until a set of cliques is obtained. These cliques are the input of subroutine \texttt{search} which applies the grow and swap heuristic, thus obtaining a set of potential max cliques (line~7). Each clique in this group corresponds to a set of vertices in $G$ with associated colors.

\emph{Algorithm 3}. \texttt{BestClique} selects the most promising clique from a list of cliques inside an augmented graph, based on several hierarchical criteria. It takes as input a graph $H$ together with a list of cliques in the complement of its augmented $k$-graph. For each clique $C$ in this list (line~1) it removes the corresponding vertices from the graph $H$ to form a residual graph $H'$ containing the uncolored vertices (line~2). To select the best clique, we apply the following hierarchical criteria: 
\begin{enumerate}
    \item Maximum size $|C|$ (line~3).
    \item Minimum number of colors used in $C$ (line~4).
    \item Minimum coloring of $H'$ (line~5)\footnote{Coloring obtained using Dsatur.}.
    \item Minimum density $H'$ (line~6).
\end{enumerate}
The first criterion (line~3) aims to maximize the number of vertices that can be colored with $k$ colors. Criterion 2 (line~4) aims to select, among cliques of the same size, those that use fewer colors. In case ties still remain, we resort to the last two criteria, they ensure that the remaining uncolored vertices can be colored with as few colors as possible (line~5) and that the density of the residual graph $H'$ has the lowest density (line~6). The algorithm returns the best clique under this hierarchical criteria (line~9). 

\begin{algorithm}
\caption{\texttt{BestClique}}
\begin{algorithmic}[1]
\Require Graph $H$, cliques $C_{\text{max}}$ in $\bar{H_k}$
\ForAll{$C$ in $C_{max}$}
    \State $H' \gets H \setminus p_H(C)$ (see Definition \ref{def:projection})
    \State $k_1 \gets |C|$ 
    \State $k_2 \gets $ colors used in $C$
    \State $k_3 \gets$ additional colors estimated by \Call{\texttt{dsatur}}{$H'$} 
    \State $k_4 \gets$ density of $H'$ 
    \State $\texttt{key}(C) \gets (k_1, -k_2, -k_3, -k_4)$
\EndFor
\State \Return lexicographic $\underset{C\in C_{\text{max}}}{\texttt{argmax}}$ $\texttt{key}(C)$
\end{algorithmic}
\end{algorithm}

Our method differs from the one presented in Section \ref{sec:reduction_maxclique} in that we do not increment the value of $k$ at each iteration. Instead, we focus on finding a coloring for the uncolored vertices at each step. 


\begin{figure}[p] 
    \centering
    \includegraphics[width=\textwidth-.1in]{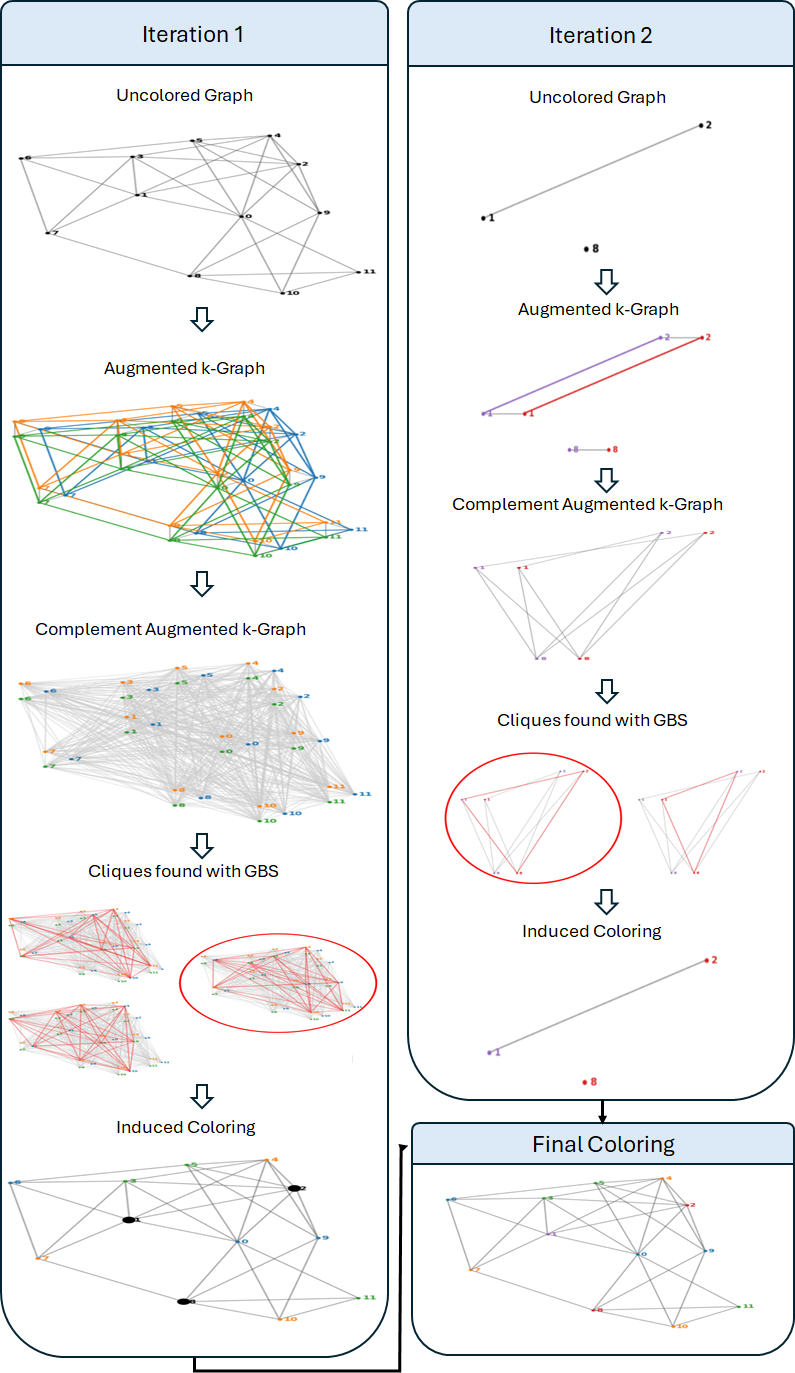}
    \caption{Illustrative example of iterative coloring with \texttt{GBSC}}
    \label{fig:flowchart}
\end{figure}


\section{Computational Results}\label{sec:results}
In this section, we evaluate the performance of our proposed graph coloring algorithm, referred to as GBSC, through a series of computational experiments. Our approach leverages the Strawberryfields library by Xanadu \cite{Bromley} for the computation of dense subgraphs and subsequent clique detection. We first assess the ability of GBSC to find near-optimal colorings for random graphs, benchmarking its performance against three established and highly effective coloring heuristics described in Section~\ref{sec:heur}: Smallest Last with Interchange (SLI) \cite{SLI}, Recursive Largest First (RLF) \cite{RLF}, and Degree of Saturation (Dsatur) \cite{Brelaz}. Subsequently, we investigate the practical applicability of GBSC in a smart-charging use case.

To ensure a rigorous comparison, we use a branch-and-price algorithm (detailed in Section \ref{sec:exact}) to determine the exact chromatic number for each graph instance. We then measure the average number of excess colors used by each method, defined as the difference between the number of colors used by the algorithm and the chromatic number, normalized by the number of graph instances.

All experiments were conducted on a high-performance computing cluster, equipped with an Intel Xeon Platinum 8558P processor and 1031 GB of RAM. The operating system was Rocky Linux 8.8 (Green Obsidian). For each experiment, we used 1 CPU core per task and a maximum runtime of 30 hours. The workload was parallelized using SLURM array jobs, with a total of parallel tasks equaling the number of instances (graphs) being treated.

\subsection{Random graph coloring}\label{sec:random_coloring}
The Erd\"{o}s-R\'enyi model is either of two closely related models for random graph generation. In one of these \cite{Erdos}, graphs with $n$ vertices are constructed by connecting vertices randomly. Each edge is included in the graph with probability $p\in [0,1]$ independently from every other edge. 

\begin{table}[b]
\centering
\begin{tabular}{c c c c c}
 & Group 1 & Group 2 & Group 3 & Group 4\\ 
 \hline
 Size &  91 & 91 & 91 & 91\Tstrut\\
 $p$ & 0.72-0.87 & 0.50-0.71 & 0.37-0.49 & 0.22-0.36\Tstrut\\
\end{tabular}
\caption{Groups of randomly generated graphs.}
\label{tab:random_graphs}
\end{table}

Following this procedure, we generated four distinct groups of random graphs. For each group, edge probabilities were sampled uniformly from a specified interval. All groups include 13 graphs for each of the following vertex counts : 10, 15, 20, 25, 30, 35, and 40. Table \ref{tab:random_graphs} summarizes the total number of instances per group (first row), along with the corresponding edge probability ranges (second row).

\begin{table*}
\begin{subtable}[h]{\linewidth}
\resizebox{\linewidth}{!}{
\begin{tabular}[b]{cccc|llll}
\toprule
& $p$ & $\chi$ & nsamples & SLI & RLF & Dsatur & GBSC \\
\midrule
Group 1 & 0.72-0.87 & 7 & 10 & 0.1 & 0.1 & 0.2 & \textbf{0}\\
&           & 9 & 10 & 0.3 & 0.4 & 0.9 & \textbf{0.1} \\
&           & 10 & 10 & 0.3 & 0.5 & 0.9 & \textbf{0.1} \\
&           & 12 & 13 & \textbf{0.38} & 0.46 & 1.0 & \textbf{0.38} \\
&           & 13 & 10 & 0.6 & 0.8 & 1.9 & \textbf{0.4} \\
\cmidrule(lr){5-8}
\multicolumn{4}{c}{Average} & 0.29 & 0.42 & 0.91 & \textbf{0.19}\\
\midrule
Group 2 & 0.50-0.71 & 5 & 10 & \textbf{0} & 0.1 & 0.2 & \textbf{0}\\
&           & 6 & 10 & 0.1 & 0.1 & 0.1 & \textbf{0}\\
&           & 7 & 15 & \textbf{0.2} & 0.4 & 0.67 & \textbf{0.2}\\
&           & 8 & 10 & 0.4 & 0.5 & 1.2 & \textbf{0.1}\\
&           & 9 & 16 & 0.63 & 0.63 & 1.5 & \textbf{0.31}\\
&           & 10 & 14 & 0.86 & 0.86 & 1.64 & \textbf{0.36}\\
\cmidrule(lr){5-8}
\multicolumn{4}{c}{Average} & 0.41 & 0.51 & 1.01 & \textbf{0.18} \\
\midrule
Group 3 & 0.37-0.49 & 4 & 13 & \textbf{0.15} & 0.38 & 0.69 & \textbf{0.15}\\
&           & 5 & 16 & \textbf{0.19} & 0.31 & 0.69 & \textbf{0.19}\\
&           & 6 & 26 & 0.35 & 0.53 & 0.96 & \textbf{0.31}\\
&           & 7 & 18 & 0.61 & 0.83 & 1.22 & \textbf{0.44}\\
\cmidrule(lr){5-8}
\multicolumn{4}{c}{Average} & 0.33 & 0.52 & 0.86 & \textbf{0.26} \\
\midrule
Group 4 & 0.22-0.36 & 3 & 18 & \textbf{0} & 0.17 & 0.39 & \textbf{0}\\
&           & 4 & 27 & \textbf{0} & 0.19 & 0.44 & 0.07\\
&           & 5 & 26 & \textbf{0.15} & 0.73 & 1.0 & 0.23\\
&           & 6 & 16 & 0.25 & 0.5 & 0.94 & \textbf{0.19}\\
\cmidrule(lr){5-8}
\multicolumn{4}{c}{Average} & \textbf{0.1} & 0.4 & 0.7 & 0.13 \\
\bottomrule
\end{tabular}
}
\end{subtable}
\caption{Average number of excess colors for random graph coloring.}
\label{tab:random_comparison}
\end{table*}

Table \ref{tab:random_comparison} reports a comparative analysis of the proposed GBSC algorithm and the three benchmark heuristics—SLI, RLF, and Dsatur—based on the average excess colors relative to the chromatic number. Each of the four groups corresponds to a distinct range of edge probabilities ($p$). Within each group, results are further subdivided by chromatic number ($\chi$), and the reported averages are computed as follows: for each chromatic number, the excess colors for a given heuristic are calculated as the difference between the number of colors used by the heuristic and the exact chromatic number (determined using the branch-and-price algorithm described in Section \ref{sec:exact}). This difference is then averaged over all graph instances within the subdivision, as specified in the \emph{nsamples} column. The best results are highlighted in bold.

In Group 1, GBSC achieves the lowest excess in all chromatic number cases except for $\chi = 12$, where SLI shows comparable performance. For Group 2, GBSC consistently obtains the best results, matching SLI in two instances and outperforming all others elsewhere. In Group 3, GBSC outperforms SLI for chromatic numbers 6 and 7, while achieving comparable performance for $\chi = 4,5$. Group 4, the most challenging set, shows more mixed results. GBSC achieves the best coloring in one instance, and ties with SLI for $\chi = 3$. SLI outperforms all other methods in the remaining cases. Nevertheless, the difference in average excess colors between GBSC and SLI across this group is 0.03, which translates to a total difference of only 3 colors, a surprisingly narrow margin considering the total of 416 colors used in the optimal colorings for this group. Overall, GBSC delivers the lowest average excess across all groups, with particularly strong performance in denser graphs (Groups 1 and 2). In sparser scenarios (Groups 3 and 4), GBSC remains competitive, frequently matching or exceeding the performance of classical heuristics.

These findings are further supported by the Win-Draw-Loss (W-D-L) statistics presented in Table \ref{tab:random_W-D-L}. The comparison is conducted at the level of individual chromatic number instances (per-row comparisons within each group), excluding aggregated group averages to prevent redundancy. Overall, compared to Dsatur or RLF, GBSC records the highest number of wins (19), with no loss or draw. When compared to SLI, GBSC secures 11 wins, suffers 2 losses, and draws 6 times. Focusing on the denser graph instances, GBSC shows a strong performance compared to the other methods. On Group 1, GBSC demonstrates a clear advantage over SLI, with four wins, one draw, and no losses. For Group 2, GBSC outperformes SLI having four wins, two draws and no losses. For the sparser graphs, SLI slightly improves its performance compared to GBSC. For Group 3, SLI secures two draws with GBSC, but is outperformed two times. For Group 4 SLI outperforms GBSC two times, with one draw and one loss. It is particularly noteworthy that GBSC excels in denser graphs, which are typically more difficult to color due to their increased edge density. These results reinforce the robustness and effectiveness of GBSC across varying graph structures and densities, and underscore its superiority over classical heuristics in reducing coloring overhead.

\begin{table}
\centering
\begin{tabular}{cc c c} 
 \hline
  & SLI & RLF & Dsatur \\ 
 \hline
 Group 1 & 4-1-0 & 5-0-0 & 5-0-0 \\
 Group 2 & 4-2-0 & 6-0-0 & 6-0-0\\ 
 Group 3 & 2-2-0 & 4-0-0 & 4-0-0\\ 
 Group 4 & 1-1-2 & 4-0-0 & 4-0-0\\
 \hline
 Overall & 11-6-2 & 19-0-0 & 19-0-0\Tstrut\\
 \hline
\end{tabular}
\caption{\centering{GBSC Win-Draw-Loss (W-D-L) compared to SLI, RLF and Dsatur for random graph coloring.}}
\label{tab:random_W-D-L}
\end{table}

\subsection{Smart-charging use case}\label{subsec:ev_charging}
Charging an electric vehicle (EV) takes a finite duration with a starting time typically pre-scheduled by the user. Each EV must be charged within a specific time interval. With the advent of pre-booking systems in smart-charging, a user can select a subset of possible charging intervals, thus offering terminal operators some flexibility in assigning EVs to charging terminals. Such subsets correspond to groups of alternatives.

In this article, the smart-charging problem is, given a set of time intervals partitioned into groups, to assign all EVs to a minimum number of terminals. 
Note that this problem is a generalization of the problem considered in \cite{aquare24}, where the problem was restricted to a single terminal, then the objective was to serve the maximum number of EVs from distinct groups. 

This problem reduces to an Interval Scheduling Problem (ISP), where tasks correspond to intervals with groups containing at most $K$ tasks. $K$ represents the maximum number of tasks allowed per group across all groups. A subset of tasks is compatible if no intervals overlap and no two intervals belong to the same group. This problem is formally defined as the Group Interval Scheduling Problem (GISP), which seeks the largest compatible subset of tasks. The GISP is NP-hard for $K\geq3$ and lacks a Polynomial-Time Approximation Scheme \cite{Nakajima82}, but is polynomially solvable for $K\leq2$ \cite{Kiel92}. The GISP also reduces to finding the MIS on a Group Interval graph. In the GISP, each task corresponds to a vertex, and any two incompatible tasks correspond to an edge. In this article, the generalized problem translates into a graph coloring problem, where each terminal corresponds to a color, and the objective is to find an assignment of all vertices with a minimum number of different colors, i.e., matching the chromatic number of the graph. 

To evaluate the scalability of our coloring method, we considered graphs with 12, 20, 24, 30, 36, and 40 vertices, using 201 instances for each size. The data used to define the graph relative to each instance has been obtained from the parisian public network of charging stations for EV~\footnote{https://belib.paris/en/home} in april-may 2007, while the groups are randomly generated with $K=4$. These instances were partitioned into four quartiles based on graph density ($\rho$). Within each quartile, results were further subdivided by chromatic number ($\chi$), excess colors were computed for all graph instances in each subdivision, and then normalized by the number of samples (\emph{nsamples}) within. Performance was benchmarked against SLI, RLF, and Dsatur in Table \ref{tab:ev_comparison}, with best results highlighted in bold. 

\begin{table*}
\begin{subtable}[b]{\linewidth}
\resizebox{\linewidth}{!}{
\begin{tabular}[b]{cccc|llll}
\toprule
& $\rho$ & $\chi$ & nsamples & SLI & RLF & Dsatur & GBSC \\
\midrule
Quartile 1 & 0.11-0.23 & 4 & 32 & 0.03 & 0.41 & 0.63 & \textbf{0}\\
&           & 5 & 106 & 0.03 & 0.29 & 0.38 & \textbf{0.01} \\
&           & 6 & 98 & \textbf{0} & 0.17 & 0.19 & 0.01 \\
&           & 7 & 39 & \textbf{0} & 0.03 & 0.15 & \textbf{0} \\
&           & 8 & 23 & \textbf{0} & \textbf{0} & 0.04 & \textbf{0} \\
\cmidrule(lr){5-8}
\multicolumn{4}{c}{Average} & 0.013 & 0.21 & 0.29 & \textbf{0.006}\\
\midrule
Quartile 2 & 0.23-0.3 & 4 & 16 & 0.13 & 0.44 & 0.75 & \textbf{0}\\
&           & 5 & 50 & \textbf{0.02} & 0.18 & 0.34 & 0.1\\
&           & 6 & 72 & \textbf{0} & 0.18 & 0.36 & 0.04\\
&           & 7 & 57 & 0.02 & 0.18 & 0.35 & \textbf{0}\\
&           & 8 & 71 & \textbf{0} & 0.01 & 0.18 & 0.03\\
&           & 9 & 14 & \textbf{0} & \textbf{0} & 0.07 & \textbf{0}\\
\cmidrule(lr){5-8}
\multicolumn{4}{c}{Average} & \textbf{0.01} & 0 .13 & 0.32 & 0.03 \\
\midrule
Quartile 3 & 0.3-0.38 & 4 & 33 & \textbf{0} & \textbf{0} & 0.06 & \textbf{0}\\
&           & 5 & 57 & 0.04 & 0.19 & 0.44 & \textbf{0}\\
&           & 6 & 65 & \textbf{0} & 0.14 & 0.29 & 0.02\\
&           & 7 & 57 & 0.02 & 0.26 & 0.49 & \textbf{0}\\
&           & 8 & 44 & \textbf{0.02} & 0.11 & 0.32 & \textbf{0.02}\\
&           & 9 & 17 & \textbf{0} & 0.06 & 0.29 & \textbf{0}\\
&           & 10 & 14 & \textbf{0} & \textbf{0} & 0.36 & \textbf{0}\\
\cmidrule(lr){5-8}
\multicolumn{4}{c}{Average} & 0.013 & 0.14 & 0.33 & \textbf{0.006} \\
\midrule
Quartile 4 & 0.38-0.83 & 4 & 30 & 0.03 & 0.03 & 0.23 & \textbf{0}\\
&           & 5 & 42 & \textbf{0} & 0.05 & 0.19 & \textbf{0}\\
&           & 6 & 52 & 0.06 & 0.13 & 0.21 & \textbf{0.02}\\
&           & 7 & 77 & 0.03 & 0.05 & 0.22 & \textbf{0}\\
&           & 8 & 60 & \textbf{0} & 0.03 & 0.05 & 0.02\\
&           & 9 & 32 & \textbf{0} & \textbf{0} & 0.22 & \textbf{0}\\
\cmidrule(lr){5-8}
\multicolumn{4}{c}{Average} & 0.02 & 0.05 & 0.18 & \textbf{0.01} \\
\bottomrule
\end{tabular}
}
\end{subtable}
\caption{Average number of excess colors for the smart-charging use case.}
\label{tab:ev_comparison}
\end{table*}

For the first quartile, GBSC consistently achieves the best performance across all chromatic classes except for $\chi=6$, where SLI slightly outperforms it. Notably, this discrepancy corresponds to a single instance where GBSC required one additional color, which is negligible relative to the 588 colors used across this class. In the second quartile, GBSC reaches optimal colorings in half of the chromatic classes, while SLI performs better in the remaining cases. RLF surpasses GBSC in only one class, also by a single color. Results for the third quartile strongly favor GBSC: it outperforms SLI in six of the seven chromatic classes, losing only in the $\chi=6$ class for a single instance and by one color. In two classes of this quartile, RLF matches GBSC’s performance. Finally, in the fourth quartile, GBSC is again the dominant method, outperforming SLI in half of the classes and matching it in two. Only in the class with $\chi=8$ does SLI achieve a marginally better results, outperforming GBSC in just one instance. It is also worth mentioning that out of the 24 chromatic classes displayed in Table \ref{tab:ev_comparison}, our method finds the optimal coloring 15 times, as compared to 13 times for SLI, and 5 times for RLF.

\begin{table}
\centering
\begin{tabular}{cc c c} 
 \hline
  & SLI & RLF & Dsatur \\ 
 \hline
Quartile 1 & 2-2-1 & 4-1-0 & 5-0-0 \\
Quartile 2 & 2-1-3 & 4-1-1 & 6-0-0\\ 
Quartile 3 & 2-4-1 & 5-2-0 & 7-0-0\\ 
Quartile 4 & 3-2-1 & 5-1-0 & 6-0-0\\
 \hline
 Overall & 9-9-6 & 18-5-1 & 24-0-0\Tstrut\\
 \hline
\end{tabular}
\caption{\centering{GBSC Win-Draw-Loss (W-D-L) compared to SLI, RLF and Dsatur for smart-charging use case.}}
\label{tab:ev_W-D-L}
\end{table}

These findings are summarized in Table \ref{tab:ev_W-D-L}, which reports the Win–Draw–Loss (W–D–L) statistics based on the comparisons in Table \ref{tab:ev_comparison}. The comparison is performed at the level of individual chromatic number instances (per-row comparisons), excluding group averages to avoid redundancy. In general, GBSC exhibits flawless performance against Dsatur. Compared to RLF, it secures 18 wins, 5 draws, and only 1 loss. Against SLI, GBSC achieves 9 wins, 9 draws, and 6 losses. At the quartile level, GBSC surpasses all benchmarks except in the second quartile, where SLI obtains 3 wins compared to 2 for GBSC and 1 draw. Overall, these results highlight the robustness of GBSC across diverse graph densities and chromatic complexities, demonstrating its superiority over classical heuristics in finding minimal colorings.

\section{Conclusion}\label{sec:conclusion}
We have introduced GBSC, a hybrid classical-quantum heuristic for the graph coloring problem. The proposed approach reformulates the coloring problem through an augmented graph construction, exploiting the correspondence between independent sets and cliques in complementary graphs. GBS is then used to generate promising candidate cliques, which are refined and incorporated into an iterative classical coloring procedure. This strategy amounts to finding a coloring incrementally, thus allowing us to save computing resources.

Extensive experiments on both random graphs and graphs arising from a smart-charging use case demonstrate that GBSC consistently achieves competitive or superior results compared to state-of-the-art heuristic methods, particularly in dense graph regimes where coloring is the most challenging. In the random graph experiments, GBSC achieved the lowest overall average excess colors, outperforming SLI, RLF, and Dsatur in most scenarios. In the EV scheduling application, GBSC often finds the optimal coloring and exhibits stable performance across a wide range of densities and chromatic numbers.

While these results highlight the potential of quantum-inspired sampling techniques for tackling combinatorial optimization problems, it is important to note that our experiments were conducted via classical simulation, which inherently limits the size of the graphs that can be studied. As photonic quantum computing hardware continues to evolve — addressing challenges such as photon loss, or single-photon detection efficiency — we anticipate that our approach could be further scaled up, unlocking new possibilities for solving larger and more complex instances.




\newpage
\addcontentsline{toc}{section}{References}
\bibliographystyle{ieeetr}
\bibliography{reference.bib}
\nocite{*}

\appendix

\section{Gaussian Boson Sampling}\label{app:GBS}
This appendix provides the mathematical background on Gaussian Boson Sampling (GBS) used throughout this work.

\subsection{Boson Sampling}
Photonic Boson Sampling \cite{Gard} involves sending a single photon Fock state to a linear interferometer (a mixture of rotation and beamsplitter gates), described by a matrix $A$, which transforms $n$ input qumodes into $n$ output qumodes. A general \emph{qumode} can be expressed as a linear combination $\lvert \psi\rangle = \sum_{n=0}^\infty c_k \lvert k\rangle$, of unitary norm. A basis state $\lvert k\rangle$ has the physical interpretation of a qumode with $k$ photons. We denote a state of $k$ single-photons arranged in $n$ qumodes as $\lvert \mathbf{s}\rangle = \lvert s_1,\ldots,s_n\rangle$, where $s_i$ is the number of photons in the $i$-th qumode, and $\sum_i s_i = k$.

\subsection{Gaussian States}
A Gaussian state is characterized only by a complex $2n\times 2n$ covariance matrix $\sigma$, and a displacement vector $d\in \mathbb{C}^n$ \cite{Ferraro}. In case of null displacement ($d=0$), the Gaussian state can be prepared from the vacuum using only single-mode squeezing followed by linear-optical interferometry. Hamilton et al. \cite{Hamilton} show that when the modes of such a Gaussian state are measured, the probability $P(\mathbf{s})$ of observing a pattern of photons $\lvert\mathbf{s}\rangle$, is given by 
\begin{align}\label{GbsProb}
P(\mathbf{s})= \frac{1}{\det(\sigma_Q)^{1/2}}\frac{\Haf(\mathfrak{A}_\mathbf{s})}{s_1!\ldots s_n!}
\end{align}
where $\sigma_Q = \sigma + \mathbf{I}_{2n}/2$, and
$$
\mathfrak{A} = \bigl[ 
\begin{smallmatrix}
0 & \mathbf{I}_n\\ \mathbf{I}_n & 0
\end{smallmatrix} 
\bigr]
(\mathbf{I}_{2n} - \sigma_Q^{-1}).
$$
The submatrix $\mathfrak{A}_\mathbf{s}$ is constructed by deleting columns and rows $(i,i+n)$ of $\mathfrak{A}$ if $s_i=0$, and by repeating them $s_i$ times in case $s_i>0$; the \emph{hafnian} is defined as:
\begin{align}\label{defHafnian}
\Haf(\mathfrak{A})=\sum_{\pi \in \PMP} \prod_{(i,j)\in \pi} \mathfrak{A}_{ij},
\end{align}
where $\PMP$ is the set of \emph{perfect matching permutations}. A perfect matching in a graph is a matching that covers every vertex of the graph, an edge subset such that each vertex in the graph is connected by an edge with another unique vertex.

The matrix $\mathfrak{A}$ can be written as $\mathfrak{A} = \bigl[\begin{smallmatrix} A & C\\ C^T & A^* \end{smallmatrix}\bigr]$. If $C=0$, the state is called \emph{pure} and Equality (\ref{GbsProb}) becomes:
\begin{align}\label{GbsPureProb}
P(\mathbf{s})= \frac{1}{\det(\sigma_Q)^{1/2}}\frac{\lvert\Haf(A_\mathbf{s})\rvert^2}{s_1!\ldots s_n!},
\end{align}
where $A_\mathbf{s}$ is the submatrix that comprises only the rows and columns (with repetitions) where photons were detected.

Several photons can be detected in a given mode ($s_i>1$). If only the location of the observed photons is relevant (not their quantity at a given vertex), we can set $s_i=1$ whenever $s_i>1$. Physically, this corresponds to using (non-number resolution) \emph{threshold detectors}, which destructively determines whether a mode contains photons or not. In particular, they \emph{click} whenever one or more photons are observed. If necessary, the number of photons in a given mode can be obtained thanks to the (number-resolving) \emph{single photon detectors} \cite{AA}. 

\subsection{GBS with Squeezed States}
For pure Gaussian states with zero displacement, we can use the equality defined in (\ref{GbsPureProb}) to devise a quantum sampling problem. First, using the Takagi-Autonne decomposition \cite{Horn_Johnson} we can write
\begin{align}\label{TakaguAutonne}
A = U\diag(\lambda_ 1,\ldots,\lambda_n)U^T,
\end{align}
where $U$ is a unitary matrix. We can then use the values $\lambda_i\in [0,1[$ to determine the single-mode squeezing parameter $r_i = \arctanh(\lambda_i)$, and use the unitary $U$ for the linear interferometer. The \emph{mean photon number} of the distribution is given by $\hat{k}=\sum \lambda_i^2/(1-\lambda_i^2)$.

To encode an arbitrary symmetric matrix $A$ on a GBS device, we rescale it by a parameter $c > 0$ to obtain a matrix $cA$  that satisfies $c\lambda_i \in [0,1[$. This parameter controls the squeezing parameters $r_i$ and the mean photon number. 

A GBS device can be programmed by the following procedure:
\begin{itemize}
    \item[1.] Compute decomposition \eqref{TakaguAutonne} of $A$ to determine $U$ and $\lambda_1,\ldots,\lambda_n$.
    \item[2.] Program linear-interferometer according to $U$.
    \item[3.] Find $c>0$ such that $\hat{k}=\sum (c\lambda_i)^2/(1-(c\lambda_i^2))$.
    \item[4.] Program squeezing gates according to $r_i=\arctanh(c\lambda_i)$.
\end{itemize}
The GBS device then samples according to 
\begin{align}\label{GBSSymMat}
P(\mathbf{s}) \propto c^k \frac{\lvert\Haf(A_\mathbf{s})\rvert^2}{s_1!\ldots s_n!}, \mbox{ where }k=\sum s_i
\end{align}

Since a binary qumode $\lvert\mathbf{s}\rangle$ can be identified with a subgraph $G[\mathbf{s}]$ of $G$, and $A_\mathbf{s}$ with its adjacency matrix, we can rewrite (\ref{GBSSymMat}) as the probability of sampling the subgraph $G[s]$:
\begin{align}\label{GBSSubGraphProb}
P(G[\mathbf{s}]) \propto \frac{\Perf(G[\mathbf{s}])^2}{s_1!\ldots s_n!},
\end{align}

\subsection{Complexity}
In computational complexity, an algorithm is regarded as efficient if its runtime and resource requirements scale at most polynomially with the input size. Using only linear optical components—such as beamsplitters, phase shifters, photodetectors, and classical feedback from detection outcomes—it can be shown that such efficiency is achievable. Indeed, the unitary operator governing the evolution of the input state (Equality \ref{TakaguAutonne}) can be efficiently decomposed into $\mathcal{O}(n^2)$ optical elements \cite{Reck}. By contrast, the calculation of matrix permanents is a $\#\text{P}$-complete problem, which is generally considered even harder than NP-complete problems. The fastest known classical algorithm, due to Ryser~\cite{Ryser}, requires $\mathcal{O}(2^k k^2)$ time. Consequently, any attempt to classically simulate boson sampling by directly computing matrix permanents would demand exponential computational resources.




\section{Strawberryfields}\label{app:strawberryfields}

This appendix summarizes the Strawberryfields subroutines employed in Algorithm~2. All functions are part of the \texttt{strawberryfields.apps} module and are used with their default parameters unless otherwise specified.

\subsection{Sampling}

The function \texttt{sample} provides synthetic GBS samples based on a given symmetric matrix $A$, and a mean number of photons $\mu$. Optionally, a detection scheme (threshold or photon-number-resolving), the fraction of photons loss, and the number of resulting samples can be specified.

In our experiments, matrix $A$ is the adjacency matrix of $\bar{H_k}$, and $\mu$ is equal to $|H|$. We generate $6\cdot|H|$ samples using threshold detection, and assume no photon loss. 

\subsection{Cliques}

The \texttt{shrink} and \texttt{search} functions operate on subgraphs derived from GBS samples.
\newline

- \texttt{shrink} iteratively removes vertices from an input subgraph until it forms a clique. At each iteration, the vertex of minimum degree relative to the current subgraph is removed, with ties resolved uniformly at random or according to vertex weights when provided. In our experiments, we use this function to find cliques in $\bar{H_k}$ from dense subgraphs sampled by GBS.
\newline

- \texttt{search} seeks larger cliques by alternating between two phases: a \emph{growth phase}, which greedily adds compatible vertices, and a \emph{plateau search} phase, which swaps vertices to escape local optima. Vertex selection during these phases can be random, degree-based, or weight-based (if specified). The procedure repeats until the specified number of iterations is reached or no further improvements are possible. We use this function to find max cliques in $\bar{H_k}$ with number of iterations equal to $|H|$

\end{document}